\documentclass[12pt,a4paper]{article}
\usepackage[utf8]{inputenc}
\usepackage{amssymb,amsfonts,amsmath,bm}
\usepackage{subfigure}
\usepackage{graphicx}
\usepackage{hyperref}
\usepackage{dsfont}
\usepackage{bbold}
\usepackage{amsthm}
\allowdisplaybreaks

\usepackage[
top    = 3.2cm,
bottom = 3.0cm,
left   = 2.8cm,
right  = 2.8cm]{geometry}

\allowdisplaybreaks
\newtheorem{thm}{Theorem}[subsection]

\newtheorem{proposition}[thm]{Proposition}

\newtheorem{remark}[thm]{Remark}
\newtheorem{theorem}[thm]{Theorem}
\DeclareMathOperator{\tr}{tr}

\DeclareMathOperator{\str}{str}
\DeclareMathOperator{\sdet}{sdet}

\begin{document}

\begin{titlepage}

\begin{center}

{\Large \bf Integrable extensions of Adler's map via Grassmann algebras}
%On the integrability of a noncommutative Adler map

\vskip 1.5cm

{{\bf P. Adamopoulou$^{\dagger}$, S. Konstantinou-Rizos$^{\ddag}$ and G. Papamikos$^{\star}$ }} 

\vskip 0.8cm

{\footnotesize
$^{\dagger}$School of Mathematical \& Computer Sciences, Heriot-Watt University, UK}
\\
{\footnotesize
$^{\ddag}$ Centre of Integrable Systems, P.G. Demidov Yaroslavl State University, Russia}
\\
{\footnotesize
$^{\star}$Department of Mathematical Sciences, University of Essex, UK}

\vskip 0.5cm

{\footnotesize {\tt E-mail: p.adamopoulou@hw.ac.uk, skonstantin84@gmail.com, g.papamikos@essex.ac.uk }}\\

\end{center}

\vskip 2.0cm

\begin{abstract}
\noindent We study certain extensions of the Adler map on Grassmann algebras $\Gamma(n)$ of order $n$. We consider a known Grassmann-extended Adler map, and assuming that $n=1$ we obtain a commutative extension of Adler's map in six dimensions. We show that the map satisfies the Yang--Baxter equation, admits three invariants and is Liouville integrable. We solve the map explicitly, viewed as a discrete dynamical system.
\end{abstract}

\hspace{.2cm} \textbf{PACS numbers:} 02.30.Ik 

\hspace{.2cm} \textbf{Mathematics Subject Classification:} 15A75, 35Q53, 16T25, 17B80, 37J70

\hspace{.2cm} \textbf{Keywords:} Yang--Baxter maps, Grassmann algebras, Liouville integrability, 

\hspace{.2cm} solutions of discrete dynamical systems, symplectic structures

\vfill

\end{titlepage}

\section{Introduction}
Since the late 80s, when the study of  set-theoretical solutions to the Yang--Baxter equation started \cite{Buchstaber, skl88} (formally proposed by Drinfeld in 1990 \cite{Drinfeld}), a lot of important results have been obtained on the relation of the Yang--Baxter equation with the field of integrable systems and various algebraic structures. For example, the relation of the Yang--Baxter equation to matrix refactorisation problems \cite{KP, KP-2009, SV}, the classification of quadrirational Yang--Baxter maps \cite{ABS-2005, PTV}, the relation of Yang--Baxter maps with Darboux transformations \cite{Sokor-Sasha, KRP, MPW}, with reflection maps in soliton theory \cite{Caudrelier, Caudrelier2}, and more recently with the theory of braces \cite{DS, Rump}.

The increasing popularity of noncommutative structures over the past few decades and their applications in the field of mathematical physics and of integrable systems motivated the noncommutative (Grassmann) extension of Darboux and B\"acklund transformations \cite{GM, XLL, XL, XL-2}. This, in turn, gave rise to the construction of Grassmann extended Yang--Baxter maps. In particular, the first examples of  such maps appeared in \cite{GKRM} in relation to the nonlinear Schr\"odinger (NLS) equation and the derivative NLS equation, in other works \cite{KRK, KRM} in relation to the KdV equation, as well as for the Boussinesq equation in \cite{KR}, where a Grassmann-extended Yang--Baxter together with associated quad-graph systems was presented.

A definition for the complete integrability in the Liouville sense of Yang--Baxter maps over Grassmann algebras is not yet available. However, starting from a Grassmann Yang--Baxter map one can obtain a hierarchy of commutative maps by fixing the number of generators of the underlying Grassmann algebra. In this paper we study the Liouville integrability of such a commutative map associated to a parametric Grassmann Adler map \cite{KRM} by restricting to the Grassmann algebra with one canonical generator. 

This paper is organised as follows: In Section \ref{YB-Refactorisations} we give all the necessary definitions and notations that make this text self-contained. In particular, we give the definition of Grassmann algebras, their grading structure, and provide the definitions of superdeterminant and supertrace for matrices with Grassmann elements. Also, we give the parametric set-theoretic Yang--Baxter equation and define its solutions over sets with Grassmann variables. We call such solutions parametric Grassmann (or supersymmetric) Yang--Baxter maps. We further explain the relation of such maps to matrix refactorisation problems. In Section \ref{Adler_map} we provide a Grassmann-extended Adler map and invariant and anti-invariant quantities. In Section \ref{Adler-Int} we give the definition of the Liouville integrability for commutative maps, we present a six-dimensional commutative extension of the Adler map, and we prove its Liouville integrability. Finally, we discuss how the iterations of this map, seen as a discrete dynamical system, can be solved exactly. In Section \ref{concl} we close we some concluding remarks and ideas for future work.

\section{Yang-Baxter equation and matrix refactorisation problems on Grassmann algebras}\label{YB-Refactorisations}

A Grassmann algebra of order $n$  denoted by $\Gamma(n)$, over a field $\mathbb{F}$ of characteristic zero,  is an associative algebra with unit $\mathbb{1}$ and $n$ generators $\theta_i$, $i=1, \ldots, n$, satisfying  
\begin{equation}
    \theta_i \theta_j + \theta_j \theta_i = 0\,.
    \label{eq:defrelat}
\end{equation}
From \eqref{eq:defrelat}, it follows that $\theta_i^2=0$ for all $i$ and that a generic element of the Grassmann algebra $\Gamma(n)$ is of the form
\begin{equation}
    f=\sum_{s\geq 0}\sum_{i_1<\cdots<i_s}f_{i_1\ldots i_s}\theta_{i_1}\cdots\theta_{i_s}.
    \label{eq:graselem}
\end{equation}
We say that those elements $\eqref{eq:graselem}$ of $\Gamma(n)$ that contain sums of products of only even number of $\theta_i$'s are \textit{even}, while those that contain sums of products of only odd number of $\theta_i$'s are called \textit{odd}. We denote the subset of all even and odd elements by $\Gamma(n)_0$ and $\Gamma(n)_1$, respectively. It follows that $\Gamma(n)$ admits a $\mathbb{Z}_2$-gradation, i.e. it can be written as  $\Gamma(n)=\Gamma(n)_0\oplus\Gamma(n)_1$ with $\Gamma(n)_i\Gamma(n)_j\subseteq\Gamma(n)_{(i+j)\mod2}$. From the $\mathbb{Z}_2-$grading it follows that $\Gamma(n)_0$ is a subalgebra of $\Gamma(n)$ and that the even elements commute with all elements of $\Gamma(n)$, while the odd elements anticommute with each other.  Hence $\Gamma(n)$ has the structure of a \textit{superalgebra} with $\Gamma(n)_0$ being the set of \textit{bosonic} elements while $\Gamma(n)_1$ the set of \textit{fermionic} elements. In what follows, unless stated otherwise, elements of $\Gamma(n)_{0}$ will be denoted by Latin letters, while elements of $\Gamma(n)_{1}$ by Greek letters, with the exception of $\lambda$ which plays the role of a spectral parameter\footnote{Objects used in this paper arise in the study of spectral problems in soliton theory.}.

The notions of the determinant and the trace of a matrix with elements in $\Gamma(n)$ are defined for square matrices, $M$, of the block-form $M=\left(
\begin{matrix}
 P & \Pi \\
 \Lambda & L
\end{matrix}\right)$.
The elements of matrices $P$ and $L$ belong in $\Gamma(n)_0$, while those of $\Pi$ and $\Lambda$ in $\Gamma(n)_1$. Matrices  $\Pi$ and $\Lambda$ are not necessarily square.  The \textit{superdeterminant} of $M$, denoted by $\sdet(M)$, is defined by:
\begin{equation*}
\sdet(M)=\det(P-\Pi L^{-1}\Lambda)\det(L^{-1})=\det(P^{-1})\det(L-\Lambda P^{-1}\Pi),
\end{equation*}
where $\det(\cdot)$ is the usual determinant of a matrix, while the \textit{supertrace}, %$\str(M)$, is defined as 
\begin{equation*}
\str(M)=\tr (P)-\tr (L),
\end{equation*}
where  $\tr(\cdot)$ is the usual trace of a matrix. For more details on Grassmann algebras, see \cite{Berezin, LieSuper}.

We are interested in solutions $S_{a, b}$ to the parametric Yang--Baxter equation
\begin{equation}\label{SYB_eq1}
S^{12}_{a,b}\circ S^{13}_{a,c} \circ S^{23}_{b,c}=S^{23}_{b,c}\circ S^{13}_{a,c} \circ S^{12}_{a,b}\,,
\end{equation}
where $a, b, c \in \mathbb{F}$ and $S_{a,b}: V_n^{k, l} \times V_n^{k, l}  \rightarrow V_n^{k, l} \times V_n^{k, l}$, with 
\begin{equation}\label{Vn}
    V_n^{k, l}:= \lbrace (\bm x,\bm \chi) \, | \, \bm x \in \Gamma(n)^{k}_0,~ \bm \chi \in  \Gamma(n)^{l}_1  \rbrace \,,
\end{equation}
and the indices $k, \,l$ are the number of even and odd Grassmann variables, respectively. The maps $S_{a,b}^{ij}: V_n^{k, l} \times V_n^{k, l} \times V_n^{k, l} \rightarrow V_n^{k, l} \times V_n^{k, l} \times V_n^{k, l}$, $i,j=1,2,3$, $i\neq j$, are defined by the following relations
\begin{equation*}
S_{a,b}^{12}=S_{a, b}\times id\,, \quad S_{b,c}^{23}=id\times S_{b, c} \,,  \quad S_{a,c}^{13}=\pi^{12} S_{a,c }^{23} \pi^{12}\,,
\end{equation*}
where $\pi^{12}$ is the permutation defined by 
$$
\pi^{12}((\bm x, \bm \chi), (\bm y, \bm \psi), (\bm z, \bm \zeta)) = ((\bm y, \bm \psi), (\bm x, \bm \chi), (\bm z, \bm \zeta))\,.
$$
We call such solutions $S_{a, b}$ to the Yang--Baxter equation \eqref{SYB_eq1} {\it parametric Grassmann Yang--Baxter maps} or {\it parametric supersymmetric Yang--Baxter maps}.

Following \cite{SV}, we define a \textit{Lax matrix} of the parametric Yang--Baxter map 
$$
S_{a, b}((\bm x, \bm \chi),(\bm y, \bm \psi))=((\bm u, \bm \xi),(\bm v, \bm \eta))
$$ 
to be a matrix $\mathcal{L}_a((\bm x,\bm \chi);\lambda)$, depending on the point $(\bm x,\bm \chi) \in V_n^{k, l}$, a parameter $a\in\mathbb{F}$ and a spectral parameter $\lambda$, such that 
\begin{equation}\label{eqLax}
\mathcal{L}_a(\bm u,\bm \xi) \mathcal{L}_b(\bm v,\bm \eta)=\mathcal{L}_b(\bm y,\bm \psi) \mathcal{L}_a(\bm x,\bm \chi)\,,
\end{equation}
where we have suppressed the dependence on $\lambda$ for convenience.

The quantity $\str \left(\mathcal{L}_b(\bm y,\bm \psi)\mathcal{L}_a(\bm x,\bm \chi) \right)$ constitutes a generating function of invariants for the map $S_{a,b}$. By an invariant we mean a function $I$ of both even and odd variables such that
\begin{equation}
I(\bm u,\bm \xi, \bm v, \bm \eta)=I(\bm x,\bm \chi, \bm y, \bm \psi).
\end{equation}
It is also possible that the map admits anti-invariants, that is functions $J$ such that
\begin{equation}
J(\bm u,\bm \xi, \bm v, \bm \eta)=-J(\bm x,\bm \chi, \bm y, \bm \psi).
\end{equation}
The product of two different anti-invariants or the square of an anti-invariant are invariants.

\section{A Grassmann-extended Adler map} \label{Adler_map}

A non-commutative extension on Grassmann algebras of the Adler map  \cite{Adler}
\begin{equation}\label{Adler map}
R_{a, b}: (x,y)   \mapsto \left( y + \frac{a-b}{x+y}, x+ \frac{b-a}{x+y} \right)
\end{equation}
first appeared in \cite{KRM}, where the authors derived the map $S_{a,b}: V_n^{1,2} \times V_n^{1,2} \rightarrow V_n^{1,2} \times V_n^{1,2} $
\begin{equation} \label{map}
S_{a, b}: \left( (x, \chi_1, \chi_2), (y, \psi_1, \psi_2) \right) \mapsto \left((u, \xi_1, \xi_2), (v, \eta_1, \eta_2) \right)
\end{equation}
with components given by 
\begin{subequations} \label{map comp}
\begin{align}
x &\mapsto u = y + \frac{a-b}{x+y-\chi_1 \psi_2} \,, \\
\chi_1 &\mapsto \xi_1 = \psi_1 - \frac{a-b}{x+y} \chi_1\,,  \\
\chi_2 &\mapsto \xi_2 =  \psi_2 \,, \\
y &\mapsto v = x - \frac{a-b}{x+y - \chi_1 \psi_2}\,,  \\
\psi_1 &\mapsto \eta_1 = \chi_1 \,,\\
\psi_2 &\mapsto \eta_2 = \chi_2 + \frac{a-b} {x+y}\psi_2 \,.
\end{align}
\end{subequations}
The map $S_{a, b}$ given in \eqref{map} arises from the re-factorisation problem 
\begin{equation} \label{refact M}
\mathcal{L}_a(u, \bm \xi) \mathcal{L}_b(v, \bm \eta) = \mathcal{L}_b(y, \bm \psi) \mathcal{L}_a(x, \bm \chi)
\end{equation}
of the following Lax matrix 
\begin{equation} \label{Lax}
\mathcal{L}_a(x, \bm \chi) = 
\begin{pmatrix}
x & 1 & 0 \\
 x^2 + \chi_1 \chi_2 -a + \lambda & x& \chi_1\\
\chi_2 & 0 & 1
\end{pmatrix},
\end{equation}
with $(x, \bm \chi) = (x, \chi_1, \chi_2) \in V_n^{1,2}$. Matrix $\mathcal{L}_a(x, \bm \chi)$ was first derived in \cite{XL} in relation to the Darboux transformation for a generalised super KdV system. We note that $\mathcal{L}_a(x, \bm \chi)$ has constant superdeterminant. Here we find certain invariants and anti-invariants of the map \eqref{map}-\eqref{map comp} which was shown to be a Yang--Baxter map in \cite{KRM}.

\begin{proposition}\label{YB n}
The map $S_{a, b}$ given in \eqref{map}-\eqref{map comp}  admits the following invariants 
\begin{align}
I_{1} &= x+ y \,, \label{I1}\\
I_{2} &= \chi_1 \chi_2 + \psi_1 \psi_2 \label{I2} \,, \\
I_{3} & = b \chi_1 \chi_2 + a \psi_1 \psi_2 + (x+y) (\chi_1 \psi_2 + \psi_1 \chi_2) \label{I4} \,,\\
I_4 &= \chi_1\psi_1 \chi_2 \psi_2 \,.
\end{align}
\end{proposition}

\begin{proof}
Invariants $I_i$, $i=1,2$, follow from $\str(\mathcal{L}_b(y, \psi_1, \psi_2) \mathcal{L}_a(x, \chi_1, \chi_2))$ and were already found in \cite{KRM}. The invariance of $I_3$ can be verified by straightforward calculations. We notice that the quantities 
$$
J_1= \chi_1 \psi_1 \quad \mbox{and} \quad J_2 = \chi_2 \psi_2
$$
are anti-invariants. Indeed, 
$$
J_1 \circ S_{a,b}=\xi_1\eta_1=\left(\psi_1 - \frac{a-b}{x+y} \chi_1\right)\chi_1=\psi_1\chi_1=-J_1\,,
$$
and similarly for $J_2$. The anti-invariants $J_1$, $J_2$ lead to the invariant $I_{4} = J_1J_2 =\chi_1\psi_1 \chi_2 \psi_2$.
\end{proof}

\begin{remark}\normalfont
In the bosonic limit, where all odd variables are zero and all even variables are in $\mathbb{F}$, map \eqref{map comp} becomes the Adler map \eqref{Adler map}.
\end{remark}

\section{Adler map over $\Gamma(1)$}\label{Adler-Int}

In this section we restrict the parametric Grassmann Yang--Baxter map $S_{a, b}$ given in \eqref{map}-\eqref{map comp} in the case where all variables are in $\Gamma(1)$. Hence, we investigate the dynamical and integrability properties of the map $S_{a, b}: V_1^{1, 2} \times V_1^{1, 2} \rightarrow V_1^{1, 2} \times V_1^{1, 2}$. We first provide the definition of Liouville integrability for maps in the commutative setting. 

\newtheorem{CompleteIntegrability}{Definition}[subsection]
\begin{CompleteIntegrability}
Let  
$Y:(x_1,...,x_{2N+M})\mapsto (u_1,...,u_{2N+M})$, $u_i=u_i(x_1,...,x_{2N+M})$, $i=1,...,2N+M$, 
be a $(2N+M)$-dimensional map. If
\begin{itemize}
	\item[\textbf{i.}] there is a Poisson bracket $\lbrace \cdot,\cdot\rbrace$ such that $\left\{x_i,x_j\right\}$ has rank $2N$ and is invariant under the map $Y$;
	\item[\textbf{ii.}] map $Y$ admits $N$ functionally independent invariants, $I_i$, such that $\left\{I_i,I_j\right\}=0$, $i,j=1,\ldots,N$;
	\item[\textbf{iii.}] there are $M$ Casimir functions, $C_i$, $i=1,\ldots,M$, i.e. $\left\{C_i,f\right\}=0$, for any arbitrary function $f=f(x_1,...,x_{2N+M})$, which are invariants of the map;
\end{itemize}
then $Y$ is said to be completely integrable or Liouville integrable \cite{Fordy, Maeda, Veselov4}.
\end{CompleteIntegrability} 

In the case of the Grassmann algebra $\Gamma(1)$ with unit $\mathbb{1}$ and generator $\theta$, such that $\theta ^2 = 0$, even elements can be expressed as $p \,\mathbb{1}$ with $p \in \mathbb{F}$, while for odd elements  we have $\rho \,\theta$, $\rho \in \mathbb{F}$. 

\begin{proposition} \label{G1 prop}
The map \eqref{map}--\eqref{map comp} over $\Gamma(1)$ becomes a commutative map $\mathbb{F}^6 \rightarrow \mathbb{F}^6$
\begin{equation}\label{YB-g1}
S_{a,b} (x, \chi_{1}, \chi_{2}, y, \psi_{1}, \psi_{2}) =(u, \xi_{1}, \xi_{2}, v, \eta_{1}, \eta_{2}),
\end{equation}
where
\begin{subequations} \label{G1 map}
\begin{align}
x \mapsto u &= y + \frac{a-b}{x+y} \,, \\
\chi_{1} \mapsto \xi_{1} &= \psi_{1} - \frac{a-b}{x+y} \chi_{1}\,,  \\
\chi_{2} \mapsto \xi_{2} &=  \psi_{2} \,, \\
y\mapsto v &= x - \frac{a-b}{x+y}\,,  \\
\psi_{1} \mapsto \eta_{1} &= \chi_{1} \,,\\
\psi_{2} \mapsto \eta_{2} &= \chi_{2} + \frac{a-b} {x+y}\psi_{2} \,.
\end{align}
\end{subequations}
Map \eqref{YB-g1}-\eqref{G1 map} is a Yang--Baxter map and admits the following invariants:
\begin{subequations}
\begin{align}
I_{1} &= x + y\,,\\
I_{2} &= \chi_{1} \chi_{2} + \psi_{1}\psi_{2},\\
I_{3} & = b \chi_{1} \chi_{2} + a \psi_{1} \psi_{2} + (x+y) (\chi_{1} \psi_{2} + \psi_{1} \chi_{2}). 
\end{align}
\end{subequations}
\end{proposition}

\begin{proof}
We obtain the map in \eqref{G1 map} by restricting all the variables in the map \eqref{map comp} in the algebra $\Gamma(1)$. Then we obtain the expressions \eqref{G1 map} by equating the coefficients of $\mathbb{1}$ and $\theta$. The Yang--Baxter property can be readily verified by straightforward substitution to the classical (commutative) Yang--Baxter equation.  Finally, the invariance of $I_1, I_2, I_3$ under the map \eqref{G1 map} can be directly verified.
\end{proof}

\begin{remark}
It can also be verified that $J_1=\chi_1\psi_1$ and $J_2=\chi_2\psi_2$ are not anti-invariants of the map \eqref{G1 map}. It also follows that $I_4=J_1J_2$ given in Proposition \eqref{G1 prop} is not an invariant.
\end{remark}

\begin{theorem}
Map \eqref{YB-g1}-\eqref{G1 map} is completely integrable in the Liouville sense.
\end{theorem}

\begin{proof}
The gradients $\nabla I_i$, $i=2,3$, are linearly independent, thus the invariants $I_i$, $i=2,3$ are functionally independent. Moreover, the latter are in involution $\{I_i,I_j\}=0$ with respect to the following Poisson bracket:
\begin{equation}
 \{\chi_{1},\chi_{2}\}=\{\psi_{1},\psi_{2}\}=1 
\end{equation}
and all the rest fundamental brackets vanish. The rank of the associated $6\times 6$ Poisson matrix $P$ is four. Additionally, for the following quantities
$$
C_1=I_1=x+y,\quad C_2=ax+by+xy(x+y) 
$$
we have that $C_i\circ S_{a,b}=C_i$, $i=1,2$, i.e. they are invariants themselves, and also $(\nabla C_i)^t P=0$, from which follows that the latter are Casimir functions. Moreover, $\{C_i,I_j\}=0$ for all possible combinations. The Casimir $C_2$ can be constructed using the involutivity of the Adler map \eqref{Adler map} and averaging of the function $(a+b)y^2+2(b-a)xy-(a+b)x^2$ over an orbit of the map. Finally, map $S_{a,b}$ preserves the Poisson bracket, which completes the proof.
\end{proof}

\begin{remark}\normalfont
We observe that the commutative map \eqref{G1 map} can be decomposed into Adler's map \eqref{Adler map} and a linear map with matrix coefficients depending on the variables that appear on Adler's map. In particular the matrix coefficients are functions of the Casimir $C_1$ and hence constant for a given orbit. This implies the solvability of the map \eqref{G1 map} if seen as discrete dynamical system.
\end{remark}

\section{Conclusions}\label{concl}
In this paper we presented a commutative six-dimensional extension of Adler's map, and studied its integrability and dynamical properties. The map was derived by restricting the parametric Grassmann Yang--Baxter map given in \eqref{map}-\eqref{map comp} in the case of the $\Gamma(1)$ Grassmann algebra. We conjecture that for each Grassmann algebra $\Gamma(n)$ the commutative map obtained, in a similar way as in the case of $\Gamma(1)$, will be integrable,  leading in such a way to a hierarchy of ($3 \times 2^n$)-dimensional commutative integrable Yang--Baxter maps for $n=1,2, \ldots$. It is interesting to define complete integrability in the Grassmann setting in such a way that the integrability of the hierarchy of these commutative maps would directly follow. We aim to extend the ideas presented in this paper to the case of entwining Yang--Baxter maps.

\section*{Acknowledgements}
This paper is submitted for the proceedings of the second international conference on Integrable Systems and Nonlinear Dynamics in Yaroslavl 19--23 October 2020. 

\noindent This work started while S.K.-R. visited Heriot-Watt University in January 2019, and the University of Essex as an International Visiting Fellowship in November 2019. S.K.-R.'s work was funded by the Russian Science Foundation (project number 20-71-10110).

\end{document}